\documentclass[journal]{IEEEtran}

\ifCLASSINFOpdf
\else
\usepackage[dvips]{graphicx}
\fi
\usepackage{url}
\usepackage{microtype}
\usepackage{cite}

\usepackage{tumcolor}
\usepackage{amsmath,amssymb,amsthm,fixmath}
\usepackage{mathtools}
\usepackage{algorithm}
\usepackage{algorithmic}
\usepackage{textcomp}
\usepackage{xcolor}
\usepackage{bm}
\usepackage{pgfplots}
\usepackage{tikz}
\usetikzlibrary{intersections,pgfplots.fillbetween,patterns,spy}
\usepackage{subcaption}
\usepackage{cuted}
\usepackage{tabularx}
\usepackage{soul}
\usepackage{hyperref}
\usepackage{balance}
\usepackage[acronym,shortcuts]{glossaries}
\hypersetup{
	colorlinks=true,       
	linkcolor=black,          
	citecolor=black,        
	filecolor=black,         
	urlcolor=black        
}
\usepackage{cleveref}

\definecolor{mylila}{RGB}{153,50,204}

\newcommand{\op}[1]{{\operatorname{#1}}}
\newcommand{\B}[1]{{\bm{#1}}}
\newcommand{\uproman}[1]{\uppercase\expandafter{\romannumeral#1}}
\newcommand{\E}{\operatorname{E}}
\newcommand{\sign}{\operatorname{sign}}

\newcommand{\intall}{\int_{-\infty}^{\infty}}

\newcommand{\fhr}{f_{\B h|\B r}}
\newcommand{\fh}{f_{\B h}}
\newcommand{\frh}{p_{\B r|\B h}}
\newcommand{\fr}{p_{\B r}}
\newcommand{\dx}{\operatorname{d}\hspace{-1.5pt}}
\newcommand{\h}{^{\operatorname{H}}}
\newcommand{\T}{^{\operatorname{T}}}
\newcommand{\inv}{^{-1}}
\newcommand{\diag}{\operatorname{diag}}
\DeclareMathOperator{\eye}{\mathbf{I}}
\newcommand{\vect}{\operatorname{vec}}
\newcommand{\tr}{\operatorname{tr}}
\newcommand{\sqrtm}{^{-\frac{1}{2}}}
\newcommand{\C}{\mathbb{C}}
\newcommand{\R}{\mathbb{R}}
\newcommand{\Qinv}{\overline{Q}}
\newcommand{\erf}{\operatorname{erf}}
\newcommand{\rr}{r_{\Re}}
\newcommand{\hr}{h_\Re}
\newcommand{\nr}{n_\Re}
\newcommand{\brr}{\B r_\Re}
\newcommand{\bhr}{\B h _\Re}

\DeclareMathOperator*{\argmin}{arg\,min}

\newacronym{mimo}{MIMO}{multiple-input multiple-output}
\newacronym{simo}{SIMO}{single-input multiple-output}
\newacronym{mse}{MSE}{mean square error}
\newacronym{cme}{CME}{conditional mean estimator}
\newacronym{pdf}{PDF}{probability density function}
\newacronym{adc}{ADC}{analog-to-digital converter}
\newacronym{mmse}{MMSE}{minimum mean square error}
\newacronym{snr}{SNR}{signal-to-noise ratio}
\newacronym{evd}{EVD}{eigenvalue decomposition}
\newacronym{crb}{CRB}{Cram\'er-Rao bound}
\newacronym{map}{MAP}{maximum a posteriori}
\newacronym{cs}{CS}{compressive sensing}
\newacronym{ls}{LS}{least squares}
\newacronym{awgn}{AWGN}{additive white Gaussian noise}
\newacronym{csi}{CSI}{channel state information}
\newacronym{ml}{ML}{maximum likelihood}
\newacronym{pmf}{PMF}{probability mass function}
\newacronym{cdf}{CDF}{cumulative distribution function}
\newacronym{rv}{RV}{random variable}
\newacronym{gmm}{GMM}{Gaussian mixture model}

\newcolumntype{Y}{>{\centering\arraybackslash}X}

\colorlet{TUMOrange80}{TUMOrange!80}%
\colorlet{TUMOrange60}{TUMOrange!60}%
\colorlet{TUMOrange50}{TUMOrange!50}%
\colorlet{TUMOrange40}{TUMOrange!40}%
\colorlet{TUMOrange20}{TUMOrange!20}%

\colorlet{TUMGreen80}{TUMBeamerGreen!80}%
\colorlet{TUMGreen60}{TUMBeamerGreen!60}%
\colorlet{TUMGreen50}{TUMBeamerGreen!50}%
\colorlet{TUMGreen40}{TUMBeamerGreen!40}%
\colorlet{TUMGreen20}{TUMBeamerGreen!20}%

\colorlet{TUMRed80}{TUMBeamerDarkRed!80}%
\colorlet{TUMRed60}{TUMBeamerDarkRed!60}%
\colorlet{TUMRed50}{TUMBeamerDarkRed!50}%
\colorlet{TUMRed40}{TUMBeamerDarkRed!40}%
\colorlet{TUMRed20}{TUMBeamerDarkRed!20}%

\newtheorem{theorem}{Theorem}
\newtheorem{lemma}{Lemma}

\begin{document}
	\bstctlcite{IEEEexample:BSTcontrol}

\title{On the Mean Square Error Optimal Estimator\\ in One-Bit Quantized Systems }

\author{Benedikt Fesl, Michael Koller, and Wolfgang Utschick, \IEEEmembership{Fellow, IEEE} \vspace{-0.7cm}
	\thanks{
	}
	\thanks{
		Benedikt Fesl and Michael Koller are with Professur f\"ur Methoden der Signalverarbeitung, Technische Universit\"at M\"unchen, 80333 M\"unchen, Germany  (e-mail: benedikt.fesl@tum.de; michael.koller@tum.de).
		
		Wolfgang Utschick is with Professur f\"ur Methoden der Signalverarbeitung, Technische Universit\"at M\"unchen, 80333 M\"unchen, Germany  (e-mail: utschick@tum.de) and with the Munich Data Science Institute (MDSI).
	}
}

%

\maketitle

\begin{abstract}
	This paper investigates the\glsunset{mse} mean square error (MSE)-optimal \glsunset{cme}conditional mean estimator (CME) in one-bit quantized systems in the context of channel estimation with jointly Gaussian inputs. 
	We analyze the relationship of the generally nonlinear \ac{cme} to the linear Bussgang estimator, a well-known method based on Bussgang's theorem. 
	We highlight a novel observation that the Bussgang estimator is equal to the \ac{cme} for different special cases, including the case of univariate Gaussian inputs and the case of multiple pilot signals in the absence of additive noise prior to the quantization. For the general cases we conduct numerical simulations to quantify the gap between the Bussgang estimator and the \ac{cme}. This gap increases for higher dimensions and longer pilot sequences.
	We propose an optimal pilot sequence, motivated by insights from the \ac{cme}, and derive a novel closed-form expression of the \ac{mse} for that case. Afterwards, we find a closed-form limit of the \ac{mse} in the asymptotically large number of pilots regime that also holds for the Bussgang estimator. 
	Lastly, we present numerical experiments for various system parameters and for different performance metrics which illuminate the behavior of the optimal channel estimator in the quantized regime.
	In this context, the well-known stochastic resonance effect that appears in quantized systems can be quantified.
\end{abstract}


\IEEEpeerreviewmaketitle

\section{Introduction}
	\begin{figure}[b]
	\onecolumn
	\centering
	\copyright \scriptsize{This work has been submitted to the IEEE for possible publication. Copyright may be transferred without notice, after which this version may no longer be accessible.}
	\vspace{-1.3cm}
	\twocolumn
\end{figure}

\IEEEPARstart{D}{igital} signal processing has great impact on modern communication systems, in which the analog signals undergo quantization through \glsunset{adc}analog-to-digital converters (ADCs) . The resolution of the \acp{adc} determines how much information is preserved after the quantization of the analog signal. In many applications, the extreme case of one-bit quantization is of particular interest, e.g., in the context of lossy compression and rate distortion theory \cite{inf_theory_book,8416707}, wireless sensor networks \cite{8645383}, or audio coding \cite{audio_coding}. In recent years, one-bit quantization gained a lot of interest in \ac{mimo} communication systems, where transmitter and receiver are possibly equipped with a large number of antennas. Since the power consumption of the \acp{adc}, which are needed for every antenna, is growing exponentially with their resolution, one-bit \acp{adc} are considered as a power-efficient solution \cite{761034}. 
A major drawback of a system with one-bit \acp{adc} is the severe information loss which diminishes the capacity \cite{4594988,10.1145/1143549.1143827,6804238,9170516,ISIT2012_mezghani_nossek_ieee.pdf,7247358} and achievable rate  \cite{Li2017,7894211}, and causes signal processing tasks like channel estimation to become challenging \cite{Ivrlac2007,8314750,7439790}.
However, in contrast to multi-bit considerations, the study of one-bit quantization allows for closed-form solutions in many applications which we discuss in the following. Thereafter, the found exact solutions can be leveraged to approximate solutions for higher resolutions.

Remarkably, it is shown in \cite{4594988,10.1145/1143549.1143827} that the capacity is not severely reduced by the coarse quantization at low \acp{snr}. The capacity was further analyzed for high \acp{snr} in \cite{6804238} and for a large number of antennas in \cite{9170516}. 
In \cite{Li2017,7894211}, the achievable uplink throughput and rate is discussed. 
In order to achieve high data rates in quantized systems, accurate estimation of the \ac{csi} is particularly important. In consideration of the highly nonlinear quantization, channel estimation is considerably more demanding, 
leading to an extensive collection of approaches in the literature, e.g., \cite{708938,8683652,9641834,Liu2021,7247358,5456454,7458830,7536730,8466602,9053749,9013211,Mo2018,Dong2021,Zhang2020,Zhu2019,Nguyen2021,9443563}.
The work in \cite{7247358} studies \ac{ls} estimation. In \cite{708938,8683652,9641834,Liu2021}, iterative \ac{ml} methods are proposed, whereas in \cite{5456454,7458830,7536730}, the \ac{map} estimator is analyzed. In \cite{8466602,5456454} the \ac{crb} is investigated. In \cite{9053749,9013211,Mo2018,7536730}, \ac{cs} is studied for channel estimation, and deep learning approaches were investigated in~\cite{Dong2021,Zhang2020,Zhu2019,Nguyen2021,9443563}.

Apart from that, the Bussgang theorem turned out to be helpful for channel estimation \cite{Bussgang}. The theorem, a special case of the more general Price theorem \cite{Price}, states that the auto-correlation and cross-correlation of a zero-mean Gaussian signal before and after it has passed through a nonlinear operation is equal up to a constant. A direct consequence of the theorem is the fact that the signal after the quantization can be decomposed in a linear fashion into a desired signal part and an uncorrelated distortion. This property was investigated for various nonlinearities \cite{bussgang_examples} and further extended to non-zero quantization thresholds \cite{9046300}. In the context of signal processing, the theorem was successfully used to design channel estimators. 
To this end, the linearized model is combined with the linear \ac{mmse} estimator \cite{9307295,Kim2018,9500125,8761531}.

Despite this extensive variety of channel estimation algorithms, there is so far no detailed discussion about the \ac{cme} for the one-bit quantization case. The \ac{cme} is well-known to be the minimizer for the \ac{mse} loss function, which is the most frequently considered criterion for the channel estimation performance. More generally, the \ac{cme} is the optimal predictor for all Bregman loss functions, of which the \ac{mse} is a special case \cite{1459065}. 
On the one hand, there are approximations of the \ac{cme}, e.g., in wideband systems with many channel taps and antennas where the quantization noise becomes approximately Gaussian \cite{7600443}, or with neural networks that are trained on unquantized data \cite{Zhang2020}. 
On the other hand, the \ac{cme} is known for the simple one-dimensional case for a Gaussian input with a single pilot observation \cite[Sec. 13.1]{inf_theory_book}, \cite{7355388}. However, to our knowledge, there exists no general discussion about the multivariate case with correlated channel entries, or with longer pilot sequences that can not be optimally decorrelated.

In \cite{8761531}, a two-stage (concatenated) \ac{mmse} estimator is proposed, where in the first stage the product of the pilot matrix and the channel is estimated with the \ac{cme}. Afterwards, it is assumed that the remaining error is Gaussian, which is unlikely to hold in general, such that the channel can be recovered by decorrelating with the (pseudo-)inverse pilot matrix. This corresponds to a second stage consisting of a linear \ac{mmse} solution. Interestingly, this approach is shown to be equal to the Bussgang estimator in the one-bit quantization case.

\textit{Contributions:} We discuss the \ac{cme} in systems with one-bit quantization for the general case of multivariate Gaussian channels in the presence of \ac{awgn} and multiple pilot observations. We present closed-form solutions for special cases and numerical computations for the general case. In that respect, we discuss the optimal pilot sequence from the perspective of the \ac{cme} that was independently suggested in \cite{Zhang2020} for learning a neural network based channel estimator. We identify cases in which the Bussgang estimator is equivalent to the \ac{cme} by proving that their closed-form \ac{mse} expressions are equal and we find asymptotic \ac{mse} limits for the large number of pilots regime. 
We then quantify the gap between the linear Bussgang estimator and the nonlinear \ac{cme} in the general case by numerical experiments and show that the gap increases for higher dimensions and longer pilot sequences in the presence of correlations between the channel entries. In this context, the well-known stochastic resonance effect \cite{mcdonnell_stocks_pearce_abbott_2008}, which allows for outperforming the asymptotic limits in the presence of noise, is numerically quantified.
We finally discuss further performance metrics such as the cosine similarity and a lower bound on the achievable rate for a matched filter. 

The remainder of this work is structured as follows. In \Cref{sec:problem_formulation}, the problem formulation is described, in \Cref{sec:bussgang_estimator}, the Bussgang estimator is summarized, and in \Cref{sec:conditional_mean_estimator}, the findings on the \ac{cme} are presented. \Cref{sec:performance_metrics} then discusses further performance metrics, \Cref{sec:num_results} presents the numerical results, and \Cref{sec:conclusion} concludes the work.
\textbf{Notation:} The (element-wise) signum function is denoted by $\sign(\cdot)$. The Kronecker product is denoted by $\otimes$ and the trace of matrix is written as $\tr(\cdot)$. The real and imaginary part of a complex scalar/vector/matrix is stated as $\Re(\cdot)$ and $\Im(\cdot)$, respectively. The (column-wise) vectorization of a matrix which stacks the columns is denoted by $\vect(\cdot)$. 
We denote the \ac{pdf} of a continuous \ac{rv} $\B x$ as $f_{\B x}$ and the \ac{pmf} of a discrete \ac{rv} $\B y$ as $p_{\B y}$. 
We denote by $\mathcal{N}_\C(\B x; \B \mu, \B C)$ the complex Gaussian distribution with mean $\B \mu$ and covariance $\B C$, evaluated at $\B x$, which yields a circularly symmetric \ac{rv} after subtracting its mean.
Similarly, the real-valued Gaussian density is denoted by $\mathcal{N}(\B x; \B \mu, \B C)$.
We denote the Gaussian error function as $\erf(x) = \frac{2}{\sqrt{\pi}}\int_0^x\exp(-t^2)\op d t$.
We denote the uniform distribution between $a$ and $b$ as $\mathcal{U}(a,b)$. The function $\angle(\cdot)$ gives the angle in $[0,2\pi)$ of a complex number.
The $k,\ell$th entry of a matrix $\B X$ is denoted by $[\B X]_{k,\ell}$.


\section{Problem Formulation}\label{sec:problem_formulation}
We consider the following generic system equation
\begin{equation}
	\B R = Q(\B Y) =  Q(\B h\B a\T + \B N) \in \mathbb{C}^{N\times M}
	\label{eq:system}
\end{equation}
where $\B R=[\B r_1, \B r_2, \dots,\B r_M]$ contains $M$ quantized observations of the vector of interest $\B h \in\mathbb{C}^N$ with the known pilot vector $\B a \in\mathbb{C}^M$ that fulfills the power constraint $\|\B a\|_2^2= M$. We denote by $\B Y$ the unquantized received signal.
Let the vector $\B h$ be a zero-mean Gaussian \ac{rv}, i.e., $\B h \sim \mathcal{N}_{\mathbb{C}}(\B 0, \B C_{\B h})$.
The desired signal is distorted with additive zero-mean Gaussian noise $\B N = [\B n_1, \B n_2, \dots, \B n_M]$ where $\B n_i\sim\mathcal{N}_{\mathbb{C}}(\B 0, \B C_{\B n})$ and is afterwards quantized with the complex-valued one-bit quantization function
\begin{equation}
	Q(\cdot) = \frac{1}{\sqrt{2}} \left(\sign(\Re(\cdot)) + \op j \sign(\Im(\cdot))\right)
	\label{eq:quantization}
\end{equation}
which is applied element-wise to the input vector/matrix. The system model in \eqref{eq:system} can be equivalently described in its (column-wise) vectorized form as
\begin{equation}
	\B r = Q(\B y) = Q(\B A \B h  + \B n) \in \mathbb{C}^{NM}
	\label{eq:system_vec}
\end{equation}
with $\B A = \B a \otimes \eye$, $\B r = \vect(\B R)$, $\B y = \vect(\B Y)$, and $\B n = \vect(\B N)$.
The system model describes, e.g., a \ac{simo} wireless communication scenario where a base station equipped with $N$ antennas receives $M$ pilot signals from a single-antenna user terminal. The \acp{adc} at the receiver have one-bit resolution, which is modeled by the quantization function in \eqref{eq:quantization}.
Although we especially focus on the application of channel estimation in this work, the analysis is not limited to this instance and might yield insights to various applications such as lossy compression \cite{8416707}, wireless sensor networks \cite{8645383}, audio coding \cite{audio_coding}, or control theory \cite{quantized_meas}, 
to name a few.

The overall goal is to estimate the signal $\B h$ in an optimal sense where the vector $\B a$ is known and the channel and noise are both zero-mean Gaussian with known covariances. The optimality depends on the metric of choice, however, the most frequently used criterion in the context of channel estimation is the \ac{mse} between the true channel $\B h$ and its estimate $\hat{\B h}$, i.e.,
\begin{equation}
	\text{MSE} = \op E [\|\B h - \hat{\B h}\|_2^2]
	\label{eq:mse}
\end{equation}
for which the \ac{cme}, defined as 
\begin{equation}
	\hat{\B h} = \op E [\B h | \B r] = \int \B h \B f_{\B h | \B r}(\B h | \B r) \op d \B h,
	\label{eq:cme_definition}
\end{equation}
is known to yield optimal estimates. 
The \ac{cme} is generally not available since it requires full knowledge of the conditional \ac{pdf} $\fhr$. Furthermore, even if the \ac{pdf} is known, the \ac{cme} does generally neither have a closed-form solution nor is it feasible to compute numerically in real-time systems. For this reason, the \ac{cme} is usually approximated by sub-optimal approaches with a practically reasonable trade-off between performance and complexity, e.g., using the linear \ac{mmse} estimator. In a jointly Gaussian setting, the linear \ac{mmse} estimator is well-known to be the \ac{cme}. However, in the quantized case, this statement is no longer true in general.
Nonetheless, for a zero-mean jointly Gaussian input, the Bussgang decomposition, based on Bussgang's theorem \cite{Bussgang}, can be leveraged to design a linear \ac{mmse} estimator in the quantized case. The Bussgang estimator is sub-optimal in general due to its linearity. However, up to now there has been no analysis of the connection between the Bussgang estimator and the \ac{cme} in the case of one-bit quantization.

In this work, we investigate the \ac{mse} optimality of the Bussgang estimator and identify cases in which the Bussgang estimator is equal to the \ac{cme} and we show when the Bussgang estimator differs from the \ac{cme}. To this end, after revising the Bussgang estimator, we derive the \ac{cme} for one-bit quantization in the multivariate case. Simulation results based on the derived quantities verify the discussion.


\section{Bussgang Estimator}\label{sec:bussgang_estimator}
In this section, we briefly revise the linear \ac{mmse} estimator based on the Bussgang decomposition which is a direct consequence of Bussgang's theorem \cite{Bussgang}.  
In particular, the Bussgang decomposition implies that the system in \eqref{eq:system_vec} can be written as a linear combination of the desired signal part and an uncorrelated distortion $\B q$ as
\begin{equation}
	\B r = Q(\B y) = \B B \B y + \B \eta = \B B\B A\B h + \B q,
	\label{eq:Bussgang_decomp}
\end{equation}
where $\B B$ is the Bussgang gain that can be obtained from the linear \ac{mmse} estimation of $\B r$ from $\B y$ as
\begin{equation}
	\B B = \B C_{\B r \B y}\B C_{\B y}\inv = \sqrt{\frac{2}{\pi}}\diag(\B C_{\B y})\sqrtm,
	\label{eq:bussgang_gain}
\end{equation}
cf. \cite[Sec. 9.2]{Papoulis1965ProbabilityRV},
and where the distortion term $\B q = \B B \B n + \B \eta$ contains both the \ac{awgn} $\B n$ and the quantization noise $\B \eta$.  
As the statistically equivalent model \eqref{eq:Bussgang_decomp} is linear, one can formulate the linear \ac{mmse} estimator 
\begin{equation}
	\hat{\B h}_{\text{buss}} = \B C_{\B h\B r} \B C_{\B r}\inv\B r.
	\label{eq:h_buss}
\end{equation}
The cross-correlation matrix between the channel and the received signal is calculated as
\begin{equation}
	\B C_{\B h\B r} = \op{E}[\B h(\B B \B A \B h + \B q)\h] = \B C_{\B h}\B A\h\B B\h
\end{equation}
which follows from the fact that the noise term $\B q$ is uncorrelated with the channel $\B h$, see \cite[Appendix A]{Li2017}. 
The auto-correlation matrix $\B C_{\B r}$ is calculated via the so-called arcsine law \cite{272490} as 
\begin{equation}
		\B C_{\B r} = \frac{2}{\pi}\left( \arcsin\left(\B\Psi\Re(\B C_{\B y})\B\Psi\right) + \op j \arcsin\left(\B\Psi\Im(\B C_{\B y})\B\Psi\right)\right)
		\label{eq:arcsine_law}
\end{equation}
where $\B\Psi = \diag(\B C_{\B y})\sqrtm$.
Further reading on the Bussgang estimator can be found in \cite{Li2017,9307295}.

In \cite{8761531}, it is discussed that the Bussgang estimator is equal to a two-stage estimator that first computes the \ac{cme} of the product $\B A \B h$ of the pilot matrix and the channel, and afterwards recovers the channel $\B h$. Since the resulting error after the first step is non-Gaussian, the second step is non-optimal from the \ac{cme} perspective. In the following, we discuss cases in which the Bussgang estimator is indeed equal to the \ac{cme}, and we show the differences between the \ac{cme} and the Bussgang estimator in the general case.


\section{Conditional Mean Estimator}\label{sec:conditional_mean_estimator}
In this section, we derive the \ac{cme} for one-bit quantization for the general system in \eqref{eq:system_vec}. Afterwards, we discuss the \ac{cme} for different simplifications, i.e., with and without \ac{awgn} or multiple pilots, and for the uni- as well as multivariate case, where these cases are differently combined to provide a general overview. For each instance, we discuss possible closed-form solutions of the \ac{cme} and its relationship to the Bussgang estimator, where we identify cases in which equality between these estimators holds. 
We start by rewriting the \ac{cme} from \eqref{eq:cme_definition} as
\begin{align}
	\E[\B h | \B r] 
	&= \int \B h \frac{\frh(\B r|\B h)\fh(\B h)}{\fr(\B r)} \dx \B h \\
	&= \frac{1}{\fr(\B r)} \int \B h \fh(\B h)  \frh(\B r|\B h) \dx \B h
		\label{eq:cme_general}
\end{align}
by making use of Bayes' theorem. It is known that $f_{\B h} = \mathcal{N}_\C(\B 0, \B C_{\B h})$. In the following, we also find expressions for the prior and conditional probabilities.
We first define the set of all vectors $\B y$ which are mapped onto a given observation $\B r$ as 
\begin{equation}
	\overline{Q}(\B r) = \left\{\B y\in\mathbb{C}^{MN}: Q(\B y) = \B r\right\}.
\end{equation} 
Integrating the \ac{pdf} $f_{\B y}$ over this set yields the prior probability of observing the respective quantized observation, i.e.,
\begin{align}
	\fr(\B r) &= \int_{\Qinv(\B r)} f_{\B y}(\B y)\dx \B y,
	\label{eq:p_r}
\end{align}
where due to the joint Gaussianity we have $f_{\B y} = \mathcal{N}_{\C}(\B 0,\B C_{\B y})$ with $\B C_{\B y} =  \B A \B C_{\B h} \B A\h + \B C_{\B n}$, cf. \eqref{eq:system_vec}.
The conditional probability $\frh(\B r | \B h)$ is found similarly, however, now the integrand changes to $\mathcal{N}_\C(\B A\B h, \B C_{\B n})$, where $\B h$ is no longer a \ac{rv}:
\begin{align}
	\frh(\B r|\B h) = \int_{\Qinv(\B r)} \mathcal{N}_\C(\B x;\B A \B h, \B C_{\B n}) \dx \B x.
	\label{eq:p_r_given_h_noisy}
\end{align}
Plugging \eqref{eq:p_r} and \eqref{eq:p_r_given_h_noisy} into \eqref{eq:cme_general} yields a generally computable expression of the \ac{cme} for the model in \eqref{eq:system_vec}.

\subsection{Numerical Integration}\label{subsec:numeric}
For the most general case, the integrals in \eqref{eq:cme_general}, \eqref{eq:p_r}, and \eqref{eq:p_r_given_h_noisy} do neither have a closed-form solution, nor can they be simplified. In this case, the \ac{cme} has to be found by numerical integration. However, the number of solutions for the discrete number of possible observations $\B r\in\C^{MN}$ grows exponentially with $4^{MN}$, not to mention the increasing complexity for each individual solution.
 Since in \eqref{eq:p_r} and \eqref{eq:p_r_given_h_noisy} we have integrals over Gaussian \acp{pdf}, there exist computationally efficient algorithms. In this paper, the numerical integration of the Gaussian density is based on an algorithm from \cite{Genz1992} which first transforms the integral into an integration over the unit hypercube by combining a Cholesky decomposition of the covariance matrix with a priority ordering of the integration variables. The transformed integral is then approximated based on a Monte Carlo method which is implemented in \cite{2020SciPy}.
In the following, we discuss special instances of the general model \eqref{eq:system_vec} that allow for a simplified computation of the \ac{cme} \eqref{eq:cme_general}.




\subsection{Univariate Case with a Single Observation}\label{subsec:univariate_cme}
In the following, the case of a scalar $N=1$ system with a single pilot observation, i.e., $M=1$ and $a=1$, is considered.
Since the quantization function in \eqref{eq:quantization} is applied independently to the real and imaginary part of its input, deriving only the real part formula of the \ac{cme} is sufficient here since the imaginary part is computed analogously.
To this end we consider $\hr\sim\mathcal{N}(0, \sigma^2/2)$, $\nr \sim \mathcal{N}(0,\eta^2/2)$, and for the ease of notation we define $\rr = \sign(\hr + \nr)$.

 Because of the symmetry of the zero-mean Gaussian density, we get $p_{\rr}(\rr) = 1/2$.
Furthermore, the integral in \eqref{eq:p_r_given_h_noisy} can be computed in closed-form to, cf. \cite{Owen1980ATO},
\begin{align}
	p_{\rr | \hr} (\rr=1 | \hr) &= \int_0^\infty \mathcal{N}(x;\hr,\eta^2/2)\dx x \\
	&= \frac{1}{2} \left( 1 + \erf\left(\frac{\hr}{\eta}\right)\right),
	\label{eq:univariate_mid1} 
	\end{align}
\begin{align}
	p_{\rr | \hr} (\rr=-1 | \hr) &= \int_{-\infty}^0 \mathcal{N}(x;\hr,\eta^2/2)\dx x \\
	&= \frac{1}{2} \left( 1 - \erf\left(\frac{\hr}{\eta}\right)\right).
	\label{eq:univariate_mid2}
\end{align}
We can summarize \eqref{eq:univariate_mid1} and \eqref{eq:univariate_mid2} as
\begin{equation}
		p_{\rr | \hr} (\rr| \hr) = \frac{1}{2} \left( 1 + \rr \erf\left(\frac{\hr}{\eta}\right)\right).
\end{equation}
Plugging this result into \eqref{eq:cme_general} yields, cf. \cite[eq. (2.7.2.4)]{Integrals_err},
\begin{align}
	\op E [\hr | \rr] &= 
	    \intall h f_{h_\Re}(h)\left( 1 + \rr\erf\left(\frac{h}{\eta}\right)\right)\dx h \\
		&= \frac{1}{\sqrt{\pi}}\frac{\sigma^2}{\sqrt{\sigma^2 + \eta^2}}\rr.
\end{align}
This can now easily be extended to the complex case with $h\sim\mathcal{N}_\C(0,\sigma^2)$, $n\sim\mathcal{N}_\C(0,\eta^2)$, and $r = Q(h+n)$ which reads as
\begin{align}
	\E[h | r]
	&=  \sqrt{\frac{2}{\pi}}\frac{\sigma^2}{\sqrt{\sigma^2 + \eta^2}}r.
	\label{eq:cme_univariate}
\end{align}
The closed-form \ac{mse} from \eqref{eq:mse} is then computed to
\begin{align}
	\operatorname{MSE}
	&= \sigma^2\left(1 - \frac{2}{\pi}\frac{\sigma^2}{\sigma^2 + \eta^2}\right).
	\label{eq:mse_univariate}
\end{align}

Interestingly, if we evaluate the Bussgang estimator in \eqref{eq:h_buss}, we get the exact same \ac{cme} and \ac{mse} expressions. To conclude, for this case the Bussgang estimator is indeed the \ac{cme}. 
Note that in the multivariate case with an identity pilot matrix $\B A = \eye$ and uncorrelated channels $\B C_{\B h} = \diag(\B c_\B h)$ as well as uncorrelated noise $\B C_{\B n} = \diag(\B c_{\B n})$, the \acp{pdf} in \eqref{eq:p_r} and \eqref{eq:p_r_given_h_noisy} can be decomposed into the product of their marginals. Thus, the \ac{cme} from \eqref{eq:cme_univariate} can be applied element-wise and is therefore also equal to the Bussgang estimator.
The same expression as in \eqref{eq:cme_univariate} was found in \cite{7355388}, more generally for non zero-mean channels. However, the Bussgang theorem from \cite{Bussgang} can not be applied for non zero-mean \acp{rv} and thus, the equality does not hold in general.

Additionally, we can evaluate the well-known linear \ac{mmse} estimator for unquantized observations where it is the \ac{cme}. There, the closed-form \ac{mse} is simply 
\begin{equation}
	\text{MSE} = \sigma^2\left(1 - \frac{\sigma^2}{\sigma^2 + \eta^2}\right).
	\label{eq:mse_unquantized}
\end{equation}
We can now compute the necessary \ac{snr} for the unquantized system to achieve the same \ac{mse} as the one-bit quantized system without \ac{awgn} which we will analyze later in the numerical experiments section.

\subsection{Multivariate Noiseless Case with a Single Observation}\label{subsec:multivariate_noiseless_cme}
In this case, we consider $M=1$ with $\B n = \B 0$. Similarly as in \Cref{subsec:univariate_cme}, it is sufficient to derive only the real part of the \ac{cme}. To this end, we consider $\brr = \sign(\bhr)$ with $\bhr\sim\mathcal{N}(\B 0, \B C_{\B h})$ where $\B C_{\B h}\in\R^{N\times N}$ is a full matrix. We abuse the notation for $\Qinv(\brr) = \{\bhr \in \R^N: \sign(\bhr) = \brr\}$.
In the following Theorem \ref{prop:cme_multivariate}, a simplified expression of the general \ac{cme} \eqref{eq:cme_general} is derived whose advantage is discussed afterwards.

\begin{theorem}\label{prop:cme_multivariate}
	The \ac{cme} for the above system can be computed element-wise as
	\begin{equation}
		\begin{aligned}
			&\left[\E[\bhr | \brr]\right]_i = \frac{1}{p_{\brr}(\brr)} 
			\sum_{n=1}^N [\B r_{\Re}]_n [\B C_{\B h}]_{i,n} \\
			&~~~~~~~~~~~~~\mathcal{N}(0;0,[\B C_{\B h}]_{n,n})
			\int_{\Qinv(\brr^n)}\mathcal{N}^n(\B x; \B 0, \B C_{\B h})\dx \B x,
			\label{eq:cme_noiseless_final}
		\end{aligned}
	\end{equation}
where $\mathcal{N}^n(\B 0, \B C_{\B h})$ is the ($N-1$)-dimensional Gaussian \ac{pdf} where the $n$th row and column of $\B C_{\B h}\inv$ is deleted. The superscript $n$ of a vector truncates its $n$th element.
\end{theorem}

\textit{Proof:} See \Cref{app:proof_cme_multivariate}.

The advantage of the integral in \eqref{eq:cme_noiseless_final} is that, in contrast to the double integral that appears in \eqref{eq:cme_general}, cf. \eqref{eq:p_r_given_h_noisy}, it can be efficiently computed by means of \cite{Genz1992} as described in \Cref{subsec:numeric} because there appear only integrals over Gaussian \acp{pdf}.
Note that the expression \eqref{eq:cme_noiseless_final} simplifies to the element-wise \ac{cme} as computed in \eqref{eq:cme_univariate} for diagonal covariances since then $[\B C_{\B h}]_{i,n}^{i\neq n} = 0$ and the $(N-1)$-dimensional integral equals $2^{-(N-1)}$ due to symmetry. This analysis can be interpreted as an asymptotic \ac{snr} analysis in which the \ac{awgn} goes to zero.

The extension to the complex-valued case is straightforward by stacking the real and imaginary parts of each complex vector and using the real-valued representation $\check{\B C}_{\B h}\in \mathbb{R}^{2N\times 2N}$ of the covariance matrix $\B C_{\B h} \in \mathbb{C}^{N\times N}$, cf. \cite[Sec. 15.4]{est_theory}, given as
\begin{align}
	\check{\B C}_{\B h} = \begin{bmatrix}\Re(\B C_{\B h}) &-\Im(\B C_{\B h}) \\ \Im(\B C_{\B h}) & \Re(\B C_{\B h})\end{bmatrix}.
\end{align}

In contrast to the univariate case in \Cref{subsec:univariate_cme}, the resulting \ac{cme} expression \eqref{eq:cme_noiseless_final} is nonlinear and no longer in closed-form, which is clearly a difference to the linear Bussgang estimator. We will present numerical evaluations for this case that reveal a performance gap between the estimators.

\subsection{Univariate Noiseless Case with Multiple Observations}\label{subsec:multiple_pilots}
In this case, we consider $N=1$ with $\B n=\B 0$, i.e., the system from \eqref{eq:system_vec} simplifies to $\B r = Q(\B a h)$ with $h\sim \mathcal{N}_\C(0,\sigma^2)$.
Note that in this case it is not sufficient to compute the real-valued \ac{cme} because of the product of channel and pilot vector. This was the reason for the two-stage approach in \cite{8761531}.
Opposite to the unquantized case, having multiple pilot observations of the same realization is helpful in the context of quantization as more information can be recovered. 
Because the amplitude of the received signal is lost after the one-bit quantizer, having varying amplitude in the pilot sequence $\B a$ has no influence on the observed signal.
 However, one can introduce phase shifts in order to increase the obtained information.

In the following, we discuss the optimal pilot sequence and the resulting \ac{cme} based on it. We start by rewriting the \ac{cme} in terms of the magnitude and phase of the channel $h = \alpha \op e^{\op j \theta}$ where, due to the circular symmetry, the phase is uniformly distributed as $\theta \sim \mathcal{U}(0,2\pi)$ and is independent of the Rayleigh distributed magnitude, cf. \cite[Appendix A.1.3]{tse_viswanath_2005}. Consequently, we obtain
\begin{align}
	\op E[ \alpha \op e^{\op j \theta} | \B r]
	 = \op E[\alpha| \B  r] \hspace{1pt}\op E[ \op e^{\op j \theta} | \B r] 
	= \op E[\alpha]\hspace{1pt} \op E[ \op e^{\op j \theta} | \B r],
	\label{eq:cme_pilots_split}
\end{align}
since the magnitude information is lost due to the one-bit quantization and its estimate is thus independent of the received signal.
Therefore, we can set the magnitude of each pilot to one without loss of generality in order to fulfill the power constraint $\|\B a\|_2^2 = M$.
It is shown in the following that a recovery of the phase information of the channel depending on the number of pilots $M$ up to a certain resolution can be achieved.

\begin{figure}[t]
	\centering
	\resizebox{1\columnwidth}{!}{
		\begin{tikzpicture}[scale=1]
			\draw [->,thick,name path=xax] (-0.5,0)--(3,0) node (xaxis) [below] {$\Re$};
			\draw [->,thick,name path=yax] (0,-0.5)--(0,3) node (yaxis) [left] {$\Im$};
			\filldraw[fill=gray,opacity=0.2,draw=white] (0,0) --  (90:3) arc(90:0:3) -- cycle;
			\draw (2,-0.1) -- (2,0.1);
			\node[below] at (2,0) {$\frac{1}{\sqrt{2}}$};
			\draw (-0.1,2) -- (0.1,2);
			\node[left] at (0,2) {$\frac{1}{\sqrt{2}}$};
			\node at (2.9,0.2) {\scriptsize $\phi_{\text{low}}$};
			\node at (0.35,3) {\scriptsize $\phi_{\text{high}}$};
			\foreach \y in {2}
			\foreach \x in {2}
			\fill[black] (\x,\y) circle (3pt);
			\node at (1.75,2.1) {\small $q_1$};
			\coordinate (O) at (0,0);
			\path (O) ++(20:2.1) node (P1) {};
			\draw[fill=TUMBeamerBlue] (P1) circle (2pt) node[below] {\scriptsize $[\B y]_1$};
		\end{tikzpicture}
		\begin{tikzpicture}[scale=1]
			\draw [->,thick,name path=xax] (-0.5,0)--(3,0) node (xaxis) [below] {$\Re$};
			\draw [->,thick,name path=yax] (0,-0.5)--(0,3) node (yaxis) [left] {$\Im$};
			\filldraw[fill=gray,opacity=0.2,draw=white] (0,0) --  (60:3) arc(60:0:3) -- cycle;
			\node at (2.9,0.2) {\scriptsize $\phi_{\text{low}}$};
			\node at (1.8,2.6) {\scriptsize $\phi_{\text{high}}$};
			\draw (2,-0.1) -- (2,0.1);
			\node[below] at (2,0) {$\frac{1}{\sqrt{2}}$};
			\draw (-0.1,2) -- (0.1,2);
			\node[left] at (0,2) {$\frac{1}{\sqrt{2}}$};
			\foreach \y in {2}
			\foreach \x in {2}
			\fill[black] (\x,\y) circle (2.5pt);
			\node at (1.75,2.1) {\small $q_1$};
			\coordinate (O) at (0,0);
			\draw[name path=L2,dashed] (O)--(60:3) node[above] {};
			
			\path (O) ++(20:2.1) node (P1) {};
			\draw[fill=TUMBeamerBlue] (P1) circle (2pt) node[below] {\scriptsize $[\B y]_1$};
			\path (O) ++(50:2.1) node (P2) {};
			\draw[fill=TUMBeamerBlue] (P2) circle (2pt) node[below] {\scriptsize $[\B y]_2$};
			\draw [->,black] (P1.north east) to [out=100,in=-20] node[right] {\scriptsize $\psi_2$} (P2.north east);
			\draw [->,black] (0,1) to [out=0,in=150] node[above] {\scriptsize $\psi_2$} (0.5,0.866);
		\end{tikzpicture}
		\begin{tikzpicture}[scale=1]
			\draw [->,thick,name path=xax] (-0.5,0)--(3,0) node (xaxis) [below] {$\Re$};
			\draw [->,thick,name path=yax] (0,-0.5)--(0,3) node (yaxis) [left] {$\Im$};
			\filldraw[fill=gray,opacity=0.2,draw=white] (0,0) --  (30:3) arc(30:0:3) -- cycle;
			\draw (2,-0.1) -- (2,0.1);
			\node[below] at (2,0) {$\frac{1}{\sqrt{2}}$};
			\draw (-0.1,2) -- (0.1,2);
			\node[left] at (0,2) {$\frac{1}{\sqrt{2}}$};
			\node at (2.9,0.2) {\scriptsize $\phi_{\text{low}}$};
			\node at (2.75,1.3) {\scriptsize $\phi_{\text{high}}$};
			\foreach \y in {2}
			\foreach \x in {2}
			\fill[black] (\x,\y) circle (2.5pt);
			\node at (1.75,2.1) {\small $q_1$};
			\coordinate (O) at (0,0);
			\draw[name path=L1,dashed] (O)--(30:3);
			\path (O) ++(20:2.1) node (P1) {};
			\draw[fill=TUMBeamerBlue] (P1) circle (2pt) node[below] {\scriptsize $[\B y]_1$};
			\path (O) ++(50:2.1) node (P2) {};
			\draw[fill=TUMBeamerBlue] (P2) circle (2pt) node[below] {\scriptsize $[\B y]_2$};
			\path (O) ++(80:2.1) node (P3) {};
			\draw[fill=TUMBeamerBlue] (P3) circle (2pt) node[below] {\scriptsize $[\B y]_3$};
			\draw [->,black] (P1.north east) to [out=105,in=-0] node[right] {\scriptsize $\psi_3$} (P3.north east);
			\draw [->,black] (0,1) to [out=0,in=120] node[above] {\scriptsize $\psi_3$} (0.866,0.5);
		\end{tikzpicture}
	}
	\caption{Illustration of the iterative restriction of the angular range of the channel's phase and determination of the sub-sector by means of $\phi_{\text{low}}(\B r)$ and $\phi_{\text{high}}(\B r)$ for the case of $M=3$ pilots. 
	}
	\label{fig:pilots_sketch2}
\end{figure}
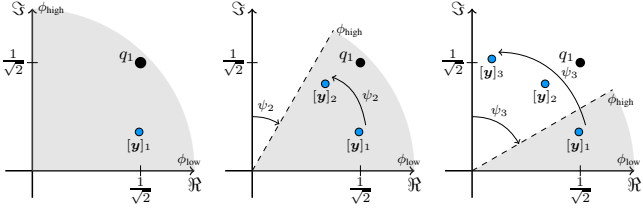

The key consideration thereby is to consecutively send pilots with an increasing phase shift and detect whether or not the quantization label is changing, which allows to narrow down the angular range of the channel's phase.
We first note that there exist four quantization points in the complex plane, each representing a quadrant with an angle of $\pi/2$. 
We further assume without loss of generality that the first pilot is sent without a phase shift, which determines in which of the four quadrants with an angular range of $\pi/2$ the channel lies. 
In turn, if the subsequent pilots have a phase shift up to $\pi/2$ the quantization label will ultimately change once the rotated channel ends up in the neighboring main quadrant.  
Building on this, an optimal pilot sequence is of the form 
\begin{equation}
	[\B a]_m = \exp(\op j\psi_m),~\psi_m\in [0,\tfrac{\pi}{2}),~m=1,\dots,M,
	\label{eq:pilot_vec_generic}
\end{equation}
which we assume to be ordered as $\psi_1 < \dots < \psi_M$.

Put differently, the pilots $m=2,...,M$ divide each of the four main sectors into $M$ circular sub-sectors if $\psi_m\neq 0~\forall m=2,...,M$. After observing a receive signal $\B r$, we would like to determine in which of the sub-sectors the channel lies. More precisely, we would like to associate a sub-sector, described by the angles $\phi_{\text{low}}(\B r)$ and $\phi_{\text{high}}(\B r)$, to every possible receive signal $\B r$ such that if we receive $\B r$, we know that the channel lies in the sub-sector described by $\phi_{\text{low}}(\B r)$ and $\phi_{\text{high}}(\B r)$.

In Fig. \ref{fig:pilots_sketch2}, we exemplarily show how to determine the sub-sector of interest for $M=3$ pilots. We therefore assume that the pilots have equidistant phase shifts in $[0, \pi/2)$ which we prove to be \ac{mse}-optimal later in Theorem \ref{theorem:pilot_opt}. The leftmost picture shows that the first unquantized receive signal $[\B y]_1 = [\B a]_1h$ lies in the positive quadrant and thus the corresponding quantized receive signal is $[\B r]_1 = Q([\B y]_1) = q_1$ which determines the main sector by means of $\phi_{\text{low}}$ and $\phi_{\text{high}}$. The middle picture shows the situation after the second pilot is sent. Since the rotated channel $[\B y]_2 = [\B a]_2h$ yields the same quantization label $[\B r]_2 = Q([\B y]_2) = q_1$, one can restrict the sector by reducing the angle $\phi_{\text{high}}$ by the angle of $\psi_2$, due to fact that if the channel would lie in the cropped area the rotated channel would yield the neighboring quantization label. The rightmost picture shows that the last receive signal is again $[\B r]_3 = q_1$, which means that the sector can be further restricted by the same argumentation. This determines the final boundary angles $\phi_{\text{low}}$ and $\phi_{\text{high}}$. In the same way as in the example of Fig. 1 for a different sequence where a change of the quantization label occurs for a certain pilot, the lower boundary angle $\phi_{\text{low}}$ is increased, ultimately yielding a sub-sector of the same size. This procedure is summarized formally in Algorithm \ref{alg:angle_boundary}.

\begin{algorithm}[t]
	\caption{Finding the boundary angles $\phi_{\text{low}}(\B r)$ and $\phi_{\text{high}}(\B r)$.}
	\label{alg:angle_boundary}
	\renewcommand{\algorithmicendfor}{\textbf{end}}
	\renewcommand{\algorithmicendwhile}{\textbf{end}}
	\renewcommand{\algorithmicendif}{\textbf{end}}
	\newcommand{\algorithmicbreak}{\textbf{break}}
	\newcommand{\BREAK}{\STATE \algorithmicbreak}
	\begin{algorithmic}[1]
		\STATE  Set the boundary angles to the quadrant boundaries: \\
			$\phi_{\text{low~}} = \angle([\B r]_1) - \frac{\pi}{4}$, $\phi_{\text{high}} = \phi_{\text{high}}^{\text{init}} = \angle([\B r]_1) + \frac{\pi}{4}$
		\FOR{$m=2$ to $M$}
		\IF{$[\B r]_m == [\B r]_{m-1}$}
		\STATE $\phi_{\text{high}} = \phi_{\text{high}}^{\text{init}} -\psi_m$
		\ELSE
		\STATE $\phi_{\text{low}} = \phi_{\text{high}} -\psi_m$
		\BREAK
		\ENDIF
		\ENDFOR
		
	\end{algorithmic}
\end{algorithm}

It remains to show that the equidistant pilot sequence minimizes the \ac{mse} in combination with the \ac{cme}.
Therefore, the \ac{cme} of the normalized channel can be obtained as
\begin{align}
	\op E[ \op e^{\op j \theta} | \B r]  &=  \int_0^{2\pi} \op e^{\op j \theta} f_{\theta | \B r}(\theta | \B r ) \dx \theta \\
	&= \frac{1}{\fr(\B r)} \int_0^{2\pi}  \op e^{\op j \theta} f_{\theta}(\theta) p_{\B r | \theta} (\B r | \theta) \dx \theta.
	\label{eq:cme_phase_bayes}
\end{align}
Due to the absence of \ac{awgn}, similarly as in \eqref{eq:p_r_given_h_noiseless}, the conditional probability $p_{\B r|\theta}$ is an indicator function 
\begin{align}
	p_{\B r| \theta}(\B r|\theta) =\begin{cases}1 &\phi_\text{low}(\B r) \leq \theta < \phi_{\text{high}}(\B r)  \\ 0 &\text{else} \end{cases}
	\label{eq:p_r_given_phi_pilots}
\end{align}
where the boundary angles $\phi_{\text{low}}(\B r)$ and $\phi_{\text{high}}(\B r)$ are determined via the above discussed procedure.
Given these preparations, we can now find the \ac{mse}-optimal pilot sequence.
\begin{theorem}\label{theorem:pilot_opt}
	The \ac{mse}-optimal pilot sequence for the above system contains equidistant phase shifts $\psi_m=\frac{\pi(m-1)}{2M}~\forall m = 1,...,M,$ such that 
	\begin{equation}
		[\B a]_m = \exp\left( \op j \left(\frac{\pi}{2M}(m-1)\right) \right),~m=1,...,M.
		\label{eq:pilot_vector}
	\end{equation}
\end{theorem}
\textit{Proof:} See Appendix \ref{app:proof_opt_pilots}.

For the optimal pilot sequence \eqref{eq:pilot_vector}, the boundary angles from \eqref{eq:p_r_given_phi_pilots} can now be compactly written as
\begin{equation}\label{eq:angle_function}
	\phi_{\text{low}}(\B r) = \frac{1}{M}\sum_{m=1}^M \angle([\B r]_m) - \frac{\pi}{4}
\end{equation}
and $\phi_{\text{high}}(\B r) = \phi_{\text{low}}(\B r) + \frac{\pi}{2M}$.
Interestingly, the same pilot sequence was proposed in \cite{Zhang2020} with the purpose to achieve a bijective mapping between a fixed dataset of channels and corresponding observations with long enough pilot sequences.

We can now find a closed-form solution for the \ac{cme} based on the optimal pilot sequence from Theorem \ref{theorem:pilot_opt}. 
\begin{theorem}\label{prop:cme_multiple_pilots}
	The \ac{cme} for the above system is computed to
	\begin{equation}
		\op E [h | \B r] = \frac{2M\sigma}{\sqrt{\pi}}\sin\left(\frac{\pi}{4M}\right)\exp\left(\op j \left(\frac{\pi}{4M} + \phi_{\text{\normalfont low}}(\B r)\right)\right).
		\label{eq:cme_multiple_pilots}
	\end{equation} 
\end{theorem}
\textit{Proof:} See Appendix \ref{app:proof_cme_multiple_pilots}.

The closed-form \ac{mse} for the \ac{cme} from \eqref{eq:cme_multiple_pilots} is then
\begin{align}
	\text{MSE}
	&= \sigma^2 \left(1 - \frac{4M^2}{\pi}\sin^2\left(\frac{\pi}{4M}\right)\right).
	\label{eq:mse_cme_multiple_pilots}
\end{align}
In contrast, the \ac{mse} for the Bussgang estimator from \eqref{eq:h_buss} for the present system is computed to
\begin{align}
	\text{MSE} 
	&= \sigma^2 \left( 1 - \frac{2}{\pi} \B a\h \B C_{\B r}\inv \B a\right).
	\label{eq:mse_buss_multiple_pilots}
\end{align}

Although the \ac{mse} terms in \eqref{eq:mse_cme_multiple_pilots} and \eqref{eq:mse_buss_multiple_pilots} seem to be different, we show their equality in the following.
\begin{theorem}\label{prop:mse_equality}
	The \ac{mse} of the Bussgang estimator from \eqref{eq:mse_buss_multiple_pilots} is equal to the \ac{mse} of the \ac{cme} from \eqref{eq:mse_cme_multiple_pilots} for an arbitrary number of pilot observations $M$.
\end{theorem}
\textit{Proof:} See Appendix \ref{app:proof_mse_equality_pilots}.

Theorem \ref{prop:mse_equality} surprisingly shows that the linear Bussgang estimator achieves the same \ac{mse} as the \ac{cme} and is thus an \ac{mse}-optimal estimator for that case. 
Interestingly, as a byproduct of the proof in Appendix \ref{app:proof_mse_equality_pilots} we have found a closed-form expression for the inverse covariance matrix $\B C_{\B r}\inv$ in \eqref{eq:mse_buss_multiple_pilots}.
Yet again, this case can be interpreted as an asymptotically high \ac{snr} analysis. 
It should be noted here that a similar proof is not possible in the presence of \ac{awgn} as shown through numerical experiments later.
With the closed-form \ac{mse} expression in \eqref{eq:mse_cme_multiple_pilots}, we can now study the asymptotically large number of pilots regime where we find the following result which, due to Theorem \ref{prop:mse_equality}, also holds for the Bussgang estimator.
\begin{theorem}\label{prop:mse_infinite_pilots}
	For $M\to \infty$ the \ac{mse} from \eqref{eq:mse_cme_multiple_pilots} is
	\begin{equation}
		\lim_{M\to\infty} \text{\normalfont MSE} = \sigma^2\left(1-\frac{\pi}{4}\right).
		\label{eq:mse_limit_multiple_pilots}
	\end{equation}
\end{theorem}
\begin{proof}
	We compute 
	\begin{align}
		\lim_{M\to\infty} M\sin\left(\frac{\pi}{4M}\right) \label{eq:lospital1}
		&= \lim_{M\to\infty }\frac{\pi}{4}\cos\left(\frac{\pi}{4M}\right)
		= \frac{\pi}{4}
	\end{align}
	where the rule of L'Hôpital was used. 
	Since this limit exists, we can conclude that 
	\begin{equation}
		\lim_{M\to\infty} M^2\sin^2\left(\frac{\pi}{4M}\right) = \frac{\pi^2}{16}\label{eq:limit_val}.
	\end{equation}
	Plugging \eqref{eq:limit_val} into \eqref{eq:mse_cme_multiple_pilots} verifies \eqref{eq:mse_limit_multiple_pilots}.
\end{proof}

\subsection{Summary}

\begin{table}[t]
	\centering
	\renewcommand{\arraystretch}{1.0}
	\begin{tabularx}{\columnwidth}{|p{0.5cm}|Y|Y|Y|Y|Y|Y|}
		\hline 
		\textbf{Sec.} &\textbf{Dim.} $N$ &\textbf{AWGN} &\textbf{Corr.} &\textbf{Pilots} $M$&\textbf{Optimal}\\
		\hline
		\ref{subsec:numeric}&$\geq1$ &yes &yes &$\geq 1$ &no \\
		\hline
		\ref{subsec:univariate_cme}&$\geq1$ &yes&no &$ 1$ &yes \\
		\hline
		\ref{subsec:multivariate_noiseless_cme}&$> 1$ &no &yes &$1$ &no \\
		\hline
		\ref{subsec:multiple_pilots}&$1$ &no &n.a. &$\geq 1$ &yes\\
		\hline	
	\end{tabularx}
	\caption{Overview of different system parameters and the \ac{mse}-optimality of the Bussgang estimator.}
	\label{tab:overview}
\end{table}

So far, we have discussed the \ac{cme} and its relationship to the Bussgang estimator for different system instances. Remarkably, in \Cref{subsec:univariate_cme} and \Cref{subsec:multiple_pilots} the Bussgang estimator is shown to be equal to the \ac{cme} which has a closed-form. This is a novel result that has not been stated in the literature so far. Because of the closed-form \ac{mse} expression of the \ac{cme}, an analysis of the large number of pilots regime was possible which resulted in a closed-form limit.
What is more, we have discussed in \Cref{subsec:multivariate_noiseless_cme} that the equality might not hold for correlated channel entries, even in the absence of \ac{awgn}. However, we have derived a computationally efficient expression for the \ac{cme} in that case. 

An overview of whether the Bussgang estimator is equal to the \ac{cme} or not for the different cases can be found in Table~\ref{tab:overview}.
Therein, we show the cases where the Bussgang estimator is \ac{mse}-optimal with respect to the dimension (Dim.) $N$, the presence/absence of \ac{awgn}, the correlation (Corr.) between channel entries and the number of pilot observations $M$. For convenience, we referenced the subsections with the corresponding derivations. The first row refers to the general case where none of the discussed special cases applies.
In the following, after introducing different performance metrics besides the \ac{mse}, we discuss the above results in numerical experiments which verify the theoretical findings.

\section{Different Performance Metrics}\label{sec:performance_metrics}
\subsection{Cosine Similarity}
The amplitude of the received signal is lost due to the one-bit quantization. For this reason, one might solely be interested in how well an estimator can recover the angle of the signal of interest. A measure for the difference of the angle $\psi$ between the true and estimated channel is the \textit{cosine similarity} which we define as
\begin{equation}
	\cos\psi = \frac{\Re(\B h\h \hat{\B h} )}{\|\B h \| \| \hat{\B h}\|}.
	\label{eq:cosine_sim}
\end{equation}
Note that the cosine similarity lies between $-1$ (opposite direction) and $1$ (same direction).

\subsection{Achievable Rate Lower Bound}
The achievable rate is of great interest in quantized systems \cite{Li2017,7894211}. 
For our considerations, we are aiming at a comparison between the \ac{cme} and the Bussgang estimator via a lower bound on the corresponding achievable rate of a respective data transmission system that is taking the \ac{csi} mismatch into account. To this end, after estimating the channel with the pilot transmission in \eqref{eq:system_vec}, the data symbol $s$ is transmitted over the same channel, i.e.,
\begin{equation}
	\B r = Q(\B h s + \B n) =  \B B \B h s + \B q
\end{equation}
where in the second equation the linearized model with Bussgang's decomposition is used where $\B q = \B B \B n + \B \eta$. We make the worst-case assumption that the aggregated noise is Gaussian, i.e., $\B q \sim \mathcal{N}_\mathbb{C}(\B 0, \B C_{\B q} = \B C_{\B r}-\B B \B C_{\B h}\B B\h )$, cf. \cite{1193803}, and use a matched filter 
\begin{equation}
	\B g_{\text{MF}}\h = \hat{\B h}\h \B B\h (\B C_{\B r}-\B B \B C_{\B h}\B B\h )\inv
\end{equation}
at the receiver. Note that the variance of the data symbol $s$ is assumed to be one without loss of generality. We further assume that the \ac{snr} is the same during pilot and data transmission.
As in \cite{7569618}, we split up the received data symbol as
\begin{align}
	\hat{s} = \underbrace{\B g_{\text{MF}}\h \B B \hat{\B h}s}_{\text{desired signal}} + \underbrace{\B g_{\text{MF}}\h\B B \B \epsilon s}_{\text{estimation error}} + \underbrace{\B g_{\text{MF}}\h\B B \B q}_{\text{noise}}
\end{align}
where $\B \epsilon = \B h - \hat{\B h}$ denotes the channel estimation error. Thus, we can evaluate a lower bound on the achievable rate, cf. \cite[eq. (20)]{7569618}, as
\begin{align}
	R &=  \op E\left[\log_2\left(1 + \frac{|\B g_{\text{MF}}\h \B B \hat{\B h}|^2}{|\B g_{\text{MF}}\h\B B \B \epsilon|^2 + \B g_{\text{MF}}\h\B B \B C_{\B q}  \B B\h \B g_{\text{MF}}} \right)\right]
		\label{eq:rate}
\end{align}
with Monte Carlo simulations.

\section{Numerical Results}\label{sec:num_results}

In this section, we perform numerical experiments that verify our theoretical derivations and allow for quantifying the performance gap between the Bussgang estimator and the \ac{cme} in the general case.
We consider a diagonal noise covariance $\B C_{\B n} = \eta^2\eye$. For the case of $N=1$, we choose without loss of generality the variance of the channel to $\sigma^2 = 1$. In the multivariate case, we randomly sample channel covariance matrices with the following procedure, cf. \cite{scikit-learn}.
We first generate a matrix $\B S$ where both the real and imaginary parts of all entries are i.i.d. as $\mathcal{U}(0,1)$. Then, we compute the \ac{evd} of $\B S\h \B S = \B V \B\Sigma \B  V\h$. Finally, we construct $\B C_{\B h} = \B V \diag(\B 1 + \B \xi)\B V\h$, where $[\B \xi]_{m} \sim\mathcal{U}(0,1)$. 
Afterwards, we normalize each covariance matrix by first dividing by its trace and afterwards multiplying by its dimension, such that $\tr(\B C_{\B h}) = N$.
For each covariance matrix we draw a single channel sample. This procedure averages out all possible structural features of the covariance matrix. For all simulations we draw a total of $10{,}000$ samples.
With the property $\|\B a\|_2^2 = M$ and $\op E [\| \B h\|^2] = N$ we define the \ac{snr} as 
\begin{equation}
	\text{SNR} = \frac{\|\B a\|_2^2\op E [\| \B h\|^2]}{\op E[\|\B n\|^2]} = \frac{1}{\eta^2}.
	\label{eq:snr}
\end{equation}
The \ac{mse} from \eqref{eq:mse} is normalized by $\op E [\| \B h\|^2] = N$.

\subsection{Univariate Case with a Single Observation}

\begin{figure}[t]
	\centering
	\includegraphics[width=1\columnwidth]{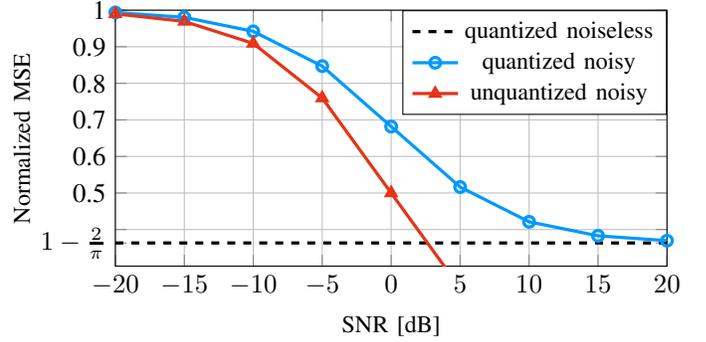}
	\caption{\ac{mse} performance of the \ac{cme} from \Cref{subsec:univariate_cme}  for $N$$=$$M$$=$$1$ for unquantized and one-bit quantized observations.}
	\label{fig:1dim_snr}
\end{figure}

%
%
%
%

We first verify the findings from \Cref{subsec:univariate_cme} with $N=M=1$ in Fig. \ref{fig:1dim_snr} where the normalized \ac{mse} is plotted against the \ac{snr}. It can be observed that the \ac{cme} converges to the noiseless \ac{mse} from \eqref{eq:mse_univariate} for $\eta=0$, which equals to $1-\frac{2}{\pi}$, for high \ac{snr} values. 
Additionally, we have plotted the well-known linear \ac{mmse} estimator for unquantized observation, where it is the \ac{cme}, from \eqref{eq:mse_unquantized}. Comparing a one-bit quantized system without \ac{awgn} with an unquantized system, we can compute their \ac{mse} intersection point by means of \eqref{eq:mse_univariate} and \eqref{eq:mse_unquantized} with the \ac{snr} definition \eqref{eq:snr} to $\text{SNR} = \frac{2}{\pi- 2} \approx 2.435 $dB which is verified by the simulation.

\subsection{Multivariate Noiseless Case with a Single Observation}

Next, we discuss the system from \Cref{subsec:multivariate_noiseless_cme} with $M=1$ and $\B n = \B 0$. As discussed above, for this case there exists no closed-form solution, but we can use a computationally efficient algorithm from \cite{Genz1992} to evaluate the integrals from \eqref{eq:p_r} and \eqref{eq:cme_noiseless_final}. In Fig. \ref{fig:dims_mse} it can be seen that the normalized \ac{mse} decreases for larger dimensions $N$ which is a consequence of the existing cross-correlations that are exploited for the estimation. Furthermore, the gap between the Bussgang estimator and the \ac{cme} increases for higher dimensions. That is, the sub-optimality of the Bussgang estimator seems to have greater impact especially in higher dimensions. Note that in the case of $N=1$ there exists the closed-form solution from before where the Bussgang estimator is exactly the \ac{cme}.

\begin{figure}[t]
	\centering
	\includegraphics[width=1\columnwidth]{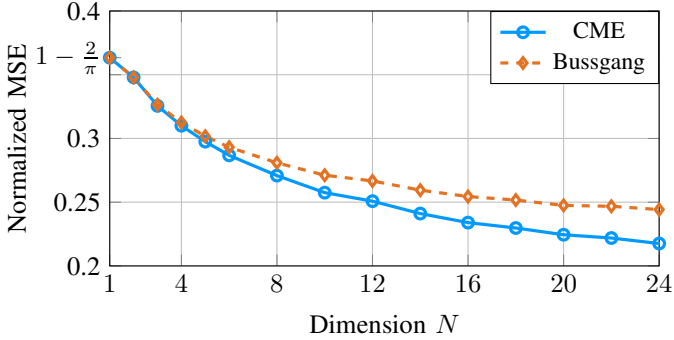}
	\caption{\ac{mse} performance comparison between the Bussgang estimator and the \ac{cme} from \Cref{subsec:multivariate_noiseless_cme} for the multivariate noiseless case with a single pilot observation.}
	\label{fig:dims_mse}
\end{figure}

\glsunset{cdf}
\begin{figure}[t]
	\centering
	\includegraphics[width=1\columnwidth]{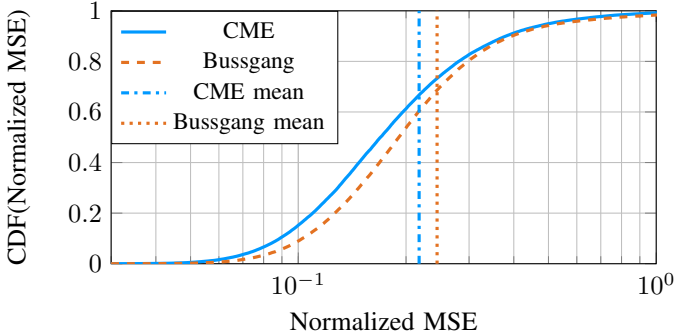}
	\caption{\ac{cdf} curves of the \ac{mse} performance comparison between the Bussgang estimator and the \ac{cme} from \Cref{subsec:multivariate_noiseless_cme} for the multivariate noiseless case with a single pilot observation for $N=24$ dimensions.}
	\label{fig:dims_cdf}
\end{figure}

\begin{figure}[t]
	\centering
	\includegraphics[width=1\columnwidth]{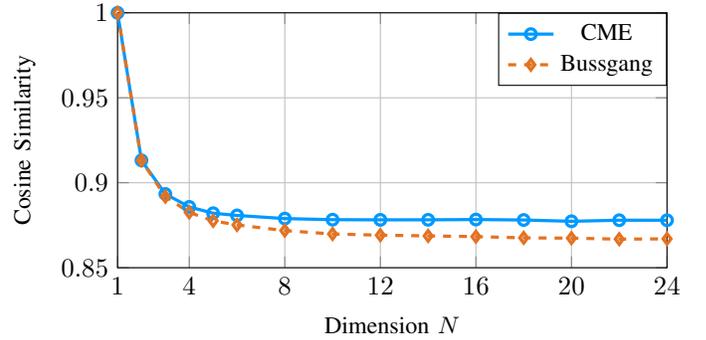}
	
	\vspace{0.2cm}
	
	\includegraphics[width=1\columnwidth]{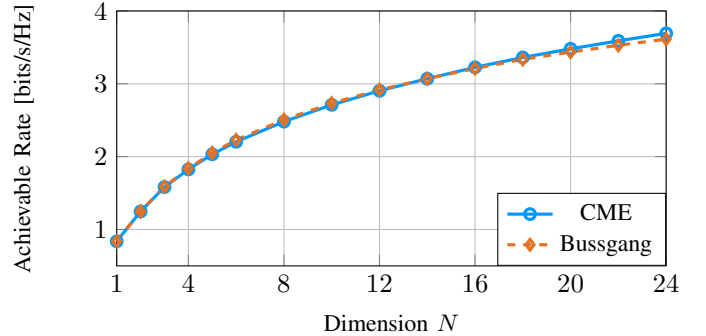}
	\caption{Cosine similarity (top) and achievable rate lower bound (bottom) comparison between the Bussgang estimator and the \ac{cme} from \Cref{subsec:multivariate_noiseless_cme} for the multivariate noiseless case with a single pilot observation.}
	\label{fig:dims_noiseless_rate}
\end{figure}

\glsreset{cdf}
In order to further investigate the performance gap between the \ac{cme} and the Bussgang estimator in this case, in Fig. \ref{fig:dims_cdf}, we evaluate the \ac{cdf} curves of the \ac{mse} of both estimators for $N=24$ dimensions. The vertical axis shows the probability that the \ac{mse} of the respective estimators is smaller than or equal to the value on the horizontal axis for each channel sample.
It can be seen that the \ac{cme} estimator is consistently better than the Bussgang estimator, especially in the region around the mean value of the \ac{mse} which we plotted for the ease of reference. In contrast, both estimators perform almost equally well for outliers, i.e., for channel samples which yield a high estimation error.

Fig. \ref{fig:dims_noiseless_rate} (top) evaluates the cosine similarity from \eqref{eq:cosine_sim} over different dimensions $N$. For an increasing number of dimensions it can be seen that the cosine similarity generally decreases. However, the \ac{cme} achieves a higher cosine similarity than the Bussgang estimator with an increasing gap for higher dimensions where a saturation is observed. This is in accordance with the results from the \ac{mse} investigations.

In Fig. \ref{fig:dims_noiseless_rate} (bottom), we analyze the achievable rate lower bound from \eqref{eq:rate} in bits/s/Hz. The lower bound is very similar for both estimators which increases from fewer than 1 bit/s/Hz for $N=1$ up to over 3 bits/s/Hz for $N>14$. 
A possible explanation for this small gap is, together with the results from the analysis of the cosine similarity, that the Bussgang estimator is able to provide a good estimate of the channel directions which seems to be especially important for the achievable rate.
Only for higher dimensions a small gap between the lower bounds can be observed which agrees with the results from the \ac{mse} analysis.

\subsection{Univariate Noiseless Case with Multiple Observations}
In this subsection, we evaluate the results from \Cref{subsec:multiple_pilots} with $N=1$ and $\B n=\B 0$. Fig. \ref{fig:pilots_noiseless} shows the normalized \ac{mse} over the number of pilots $M$. First, it can be seen that the \ac{mse} of the estimators in \eqref{eq:h_buss} and \eqref{eq:cme_multiple_pilots} are indeed equal for the conducted simulation, verifying Theorem \ref{prop:mse_equality}. Second, the \ac{mse} converges from $1-\frac{2}{\pi}$ for $M=1$ to $1-\frac{\pi}{4}$, which is the limit for $M\to \infty$ as shown in Theorem \ref{prop:mse_infinite_pilots}, already at a moderate number of $M=10$ pilots.

\subsection{General Case}

\begin{figure}[t]
	\centering
	\includegraphics[width=1\columnwidth]{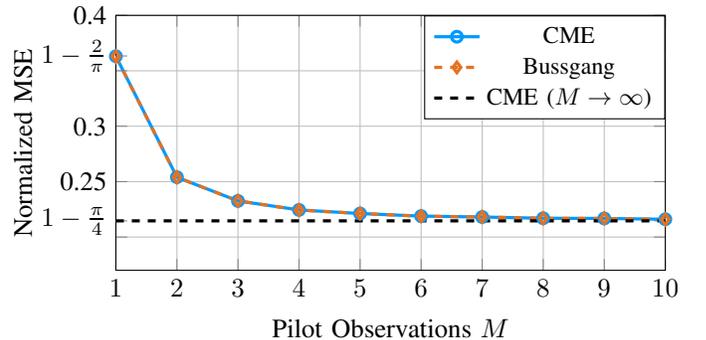}
	\caption{\ac{mse} performance comparison between the Bussgang estimator and the \ac{cme} from \Cref{subsec:multiple_pilots} for the univariate noiseless case with multiple pilot observations.}
	\label{fig:pilots_noiseless}
\end{figure}

\begin{figure}[t]
	\centering
	\includegraphics[width=1\columnwidth]{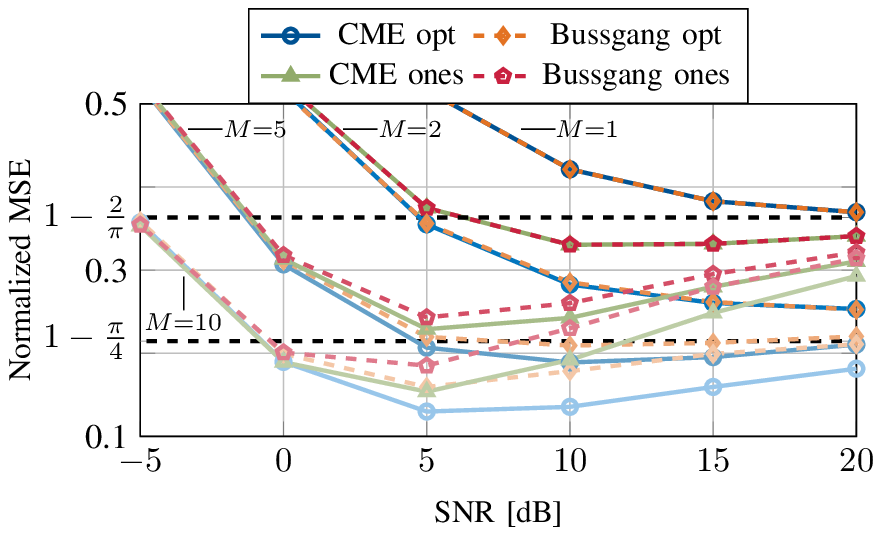}
	\caption{\ac{mse} performance comparison between the Bussgang estimator and the \ac{cme} from \Cref{subsec:numeric} for the univariate case with multiple pilot observations.}
	\label{fig:pilots_noisy}
\end{figure}

In the following, we investigate more general cases which are evaluated with numerical integration as described in \Cref{subsec:numeric}. 
First, we consider the univariate case of $N=1$ for different numbers of pilot observations ($M\in\{1,2,5,10\}$) and with \ac{awgn} for varying \acp{snr} in Fig. \ref{fig:pilots_noisy} where we compare the derived optimal pilot sequence from \eqref{eq:pilot_vector}, labeled ``opt'', with an all-ones pilot sequence $\B a=\B 1$, labeled ``ones'' which was used occasionally in recent works \cite{8683652,8466602}.
For convenience, we have also plotted the closed-form solutions for the noiseless case with $M=1$ and $M\to\infty$ as dashed lines. As expected, for $M=1$, the \ac{cme} equals the Bussgang estimator and converges to the closed-form \ac{mse} $1-\frac{2}{\pi}$ for high \ac{snr} values. Interestingly, for an increasing number of pilot observations there is an increasing gap between the Bussgang estimator and the \ac{cme}, especially in the mid to high \ac{snr} regime. What is more, the optimal pilot sequence achieves a lower \ac{mse} than the all-ones pilot sequence with a large gap especially for high \ac{snr} values, validating the analysis in \Cref{subsec:multiple_pilots} and indicating that the optimal pilot sequence performs well even in the case of finite \ac{snr} values and for the sub-optimal Bussgang estimator.
In contrast, at low \ac{snr} values, both estimators perform almost equally well which can be explained by the dominating effect of the \ac{awgn} in this regime.

It can also be seen that for higher numbers of pilots the \ac{mse} has a minimum at a certain \ac{snr} after which it increases again. What is more, the asymptotic limit of $M\to\infty$ for infinite \ac{snr} ($\B n\to\B 0$) can be outperformed with a finite number of pilots and finite \ac{snr}. These observations, seeming counterintuitive at first, are a consequence of the \textit{stochastic resonance} effect \cite{mcdonnell_stocks_pearce_abbott_2008}. This is a well-known effect related to quantization where noise can improve the performance of a system with quantization. Hence, this analysis allows for a numerical quantification of the stochastic resonance effect, which is the reason why the asymptotic limit can be outperformed and thus also describes the effect of the increasing \ac{mse} after a certain \ac{snr} value.
Interestingly, the influence of the stochastic resonance increases for higher numbers of pilot observations and the optimal pilot sequence from \eqref{eq:pilot_vector} is less prone to the effect of an increasing \ac{mse} in the high \ac{snr} regime than the all-ones sequence.

\begin{figure}[t]
	\centering
	\includegraphics[width=1\columnwidth]{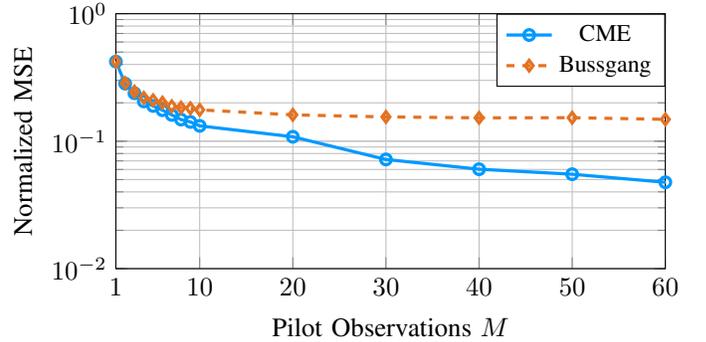}
	\caption{\ac{mse} performance comparison between the Bussgang estimator and the \ac{cme} from \Cref{subsec:numeric} for the univariate case with SNR = 10dB.}
	\label{fig:pilots_10dB}
\end{figure}

%
%
%

Next, in Fig. \ref{fig:pilots_10dB}, we analyze the behavior of the Bussgang estimator and the \ac{cme} for $N=1$ and an increasing number of pilots $M$ at a fixed \ac{snr} of 10dB. Similar as described before, the performance gap increases for higher numbers of pilots. Whereas the \ac{mse} of the Bussgang estimator saturates with a high error floor, the \ac{mse} of the \ac{cme} decreases further for higher numbers of pilots.

Finally, in Fig. \ref{fig:ndim_snr}, we depict the \ac{mse} for the case of multiple dimensions with \ac{awgn} over the \ac{snr} for correlated channel entries. For convenience, we have also plotted the closed-form solutions for $N=1$ with and without \ac{awgn}. Due to the exponentially increasing complexity of the numerical integration, we were only able to plot until $N=3$ to show the tendency of the results. First, an offset of the normalized \ac{mse} for higher dimensions can be observed which seems to decrease with the number of dimensions. Moreover, the Bussgang estimator is very close to the \ac{cme} which is in accordance with the prior results for low dimensions. The curve for $N=2$ ($N=3$) intersects with the noiseless case of $N=1$ at about 15dB (12dB) \ac{snr}.

\begin{figure}[t]
	\centering
	\includegraphics[width=1\columnwidth]{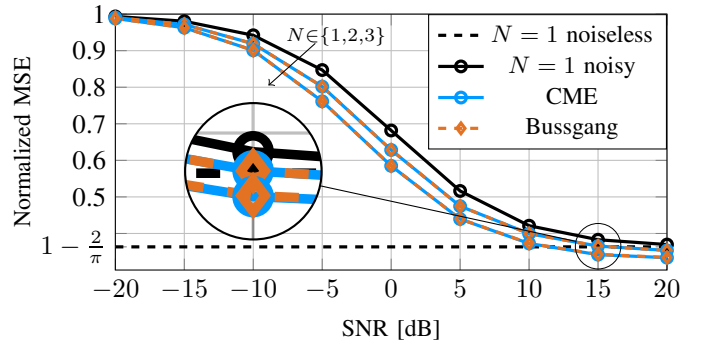}
	\caption{\ac{mse} performance comparison between the Bussgang estimator and the \ac{cme} from \Cref{subsec:numeric} for the multivariate case with a single pilot observation.}
	\label{fig:ndim_snr}
\end{figure}

\section{Conclusion and Future Work}\label{sec:conclusion}
In this work, we have investigated the \ac{cme} for one-bit quantized systems in a jointly Gaussian setting. We derived a novel closed-form solution for multiple pilot observations and a transformation into a computationally efficient expression where only Gaussian integrals have to be evaluated for the case without \ac{awgn}. 
Additionally, we have proposed an \ac{mse}-optimal pilot sequence which we motivated by the structure of the \ac{cme}. It turned out that a successive shift of the transmit signals' angle is optimal which was already proposed in prior work with a different motivation in the context of deep learning. 

We analyzed the well-known Bussgang estimator and showed the equivalence of the Bussgang estimator and the \ac{cme} in the univariate case with a single pilot observation, which extends to the multivariate case when there is no correlation between the channel or noise entries. In addition, we showed that the closed-form \ac{mse} terms of the Bussgang estimator and the \ac{cme} are equal in the case of multiple pilots without \ac{awgn}. To this end, we found a closed-form expression for the inverse of the covariance matrix of the quantized receive signal which is derived with the arcsine law. 
For the general case where numerical integration is necessary, we found that the performance gap between the Bussgang estimator and the \ac{cme} increases for higher dimensions, longer pilot sequences, or higher \ac{snr} values. 
The findings of this work may facilitate extensions to more general input distributions, e.g., \acp{gmm}, different nonlinearities such as the clipper function instead of the one-bit quantization, or to quantizers with higher resolutions.

\newcounter{MYtempeqncnt}
\begin{figure*}[!t]
	\normalsize
	\setcounter{MYtempeqncnt}{\value{equation}}
	\setcounter{equation}{66}
	
	\begin{align}
		&[\B T]_{n,n} 
		= \left.\begin{cases}  
			M - \frac{M}{2}(2 - \frac{2}{M})  & 1 < n < M \\
			M - \frac{M}{2}(1 - \frac{1}{M} + \op j \frac{1}{M}) + \op j\frac{M}{2}(1-\frac{M-1}{M} + \op j \frac{M-1}{M}) & n= 1 \\
			M - \frac{M}{2}(1-\frac{1}{M} - \op j \frac{1}{M}) - \op j \frac{M}{2}(1 - \frac{M-1}{M} + \op j \frac{1-M}{M}) & n=M
		\end{cases}
		\right\}
		= 1.
		\label{eq:Cr_inv_diag}
		\\
		&[\B T]_{m,n}^{m\neq n} 
		= \left.\begin{cases}  
			M(1- \frac{|m-n|}{M} +\op j \frac{m-n}{M}) - \frac{M}{2}(2 - \frac{|m-n-1| + |m-n+1|}{M} + 2\op j \frac{m-n}{M})  & 1 < n < M \\
			M(1 - \frac{|1-n|}{M} + \op j \frac{1-n}{M}) -\frac{M}{2} ( 1 - \frac{|2-n|}{M} + \op j \frac{2-n}{M}) +\op j \frac{M}{2} (1 - \frac{|M-n|}{M} + \op j \frac{M-n}{M})& n= 1 \\
			M(1 - \frac{|M-n|}{M} + \op j \frac{M-n}{M}) - \frac{M}{2}(1 - \frac{|M-n-1|}{M} + \op j\frac{M-n-1}{M}) - \op j \frac{M}{1} (1 - \frac{|1-n|}{M} + \op j \frac{1-n}{M})  & n=M
		\end{cases}
		\right\}
		= 0.
		\label{eq:Cr_inv_offdiag}
	\end{align}
	\setcounter{MYtempeqncnt}{46}
	\setcounter{equation}{\value{MYtempeqncnt}}
	\hrulefill
\end{figure*}

\appendix 
\subsection{Proof of Theorem \ref{prop:cme_multivariate}}\label{app:proof_cme_multivariate}
\begin{proof}
First, it can be observed that due to the absence of noise the expression in \eqref{eq:p_r_given_h_noisy} simplifies to an indicator function
\begin{equation}
	p_{\brr | \bhr} (\brr | \bhr) = \begin{cases} 1,~~\bhr\in\Qinv(\brr) \\ 0,~~\bhr\not\in\Qinv(\brr)\end{cases}
	\label{eq:p_r_given_h_noiseless}
\end{equation}
which means that the probability of receiving the quantized observation $\brr$ is one if the underlying channel lies in $\Qinv(\B r_\Re)$ since the quantization is deterministic. 
With this, we get
\begin{align}
	\E[\bhr | \brr]	&= \frac{1}{p_{\brr}(\brr)} \int_{\Qinv(\brr)} \bhr f_{\bhr}(\bhr)\dx \bhr.
	\label{eq:cme_noiseless}
\end{align}

Second, we can interpret the integral in \eqref{eq:cme_noiseless} as the (scaled) mean of a truncated multivariate Gaussian, i.e.,
\begin{align}\label{eq:integral_as_mean}
	\int_{\Qinv(\brr)} \bhr f_{\bhr}(\bhr) \dx \bhr 
	=p_{\brr}(\brr)\op E[\tilde{\B h}_\Re]
\end{align}
where $\tilde{\B h}_\Re\sim f_{\tilde{\B h}_\Re}$ is a truncated Gaussian \ac{rv} with \ac{pdf}
\begin{equation}
	f_{\tilde{\B h}_\Re}(\tilde{\B h}_\Re) = \begin{cases}\frac{1}{p_{\brr}(\brr)} \mathcal{N}(\tilde{\B h}_\Re;\B 0, \B C_{\B h}) & \tilde{\B h}_\Re\in \Qinv(\brr) 
		\\ 0 & \tilde{\B h}_\Re\not\in \Qinv(\brr)\end{cases}
\end{equation}
and $p_{\brr}(\brr)$ is the normalization such that $f_{\tilde{\B h}_\Re}$ is a valid \ac{pdf} that integrates to one.
In \cite{Tallis1961,Cartinhour1990}, the mean of a one-sided truncated Gaussian \ac{rv} is derived, which in our case is computed to
\begin{align}
		&p_{\brr}(\brr)\op E [[\tilde{\B h}_{\Re}]_i] 
		\label{eq:truncated_Gaussian_mean}
		=
		\\ &\sum_{n=1}^N [\B r_{\Re}]_n [\B C_{\B h}]_{i,n}\mathcal{N}(0;0,[\B C_{\B h}]_{n,n})
		\int_{\Qinv(\brr^n)}\mathcal{N}^n(\B x; \B 0, \B C_{\B h})\dx \B x.
		\nonumber
\end{align}
Plugging \eqref{eq:truncated_Gaussian_mean} into \eqref{eq:integral_as_mean} together with \eqref{eq:cme_noiseless} yields the expression in \eqref{eq:cme_noiseless_final} and thus finishes the proof.
\end{proof}

\subsection{Proof of Theorem \ref{theorem:pilot_opt}}\label{app:proof_opt_pilots}
\begin{proof}
	We first observe that with the procedure in Algorithm \ref{alg:angle_boundary} the angle domain can be split into at most $4M$ regions by choosing $\psi_m \neq 0$ $\forall m=2,...,M$ because any of the four quadrants can be divided into $M$ distinct regions which are identified via \eqref{eq:p_r_given_phi_pilots}. 
	Now, we write the \ac{cme} as $\hat{h} = \hat{\alpha}\op e^{\op j \hat{\theta}}$ such that we can reformulate the \ac{mse} as
	\begin{align}
		\op E[|\alpha \op e^{\op j \theta} - \hat{\alpha}\op e^{\op j \hat{\theta}}|^2] 
		&= \op E[\alpha^2 + \hat{\alpha}^2 - 2\alpha \hat{\alpha} \cos(\theta - \hat{\theta})]
		\label{eq:mse_cosine_law}
	\end{align}
	by using the cosine law. Since we are solely interested in the phase estimate for the design of the optimal pilot sequence and the cosine is monotonically decreasing between $[0,\pi/2)$, the minimization of the \ac{mse} between the phases $\op E[(\theta - \hat{\theta})^2]$ is equivalent to the minimization of the \ac{mse} from \eqref{eq:mse_cosine_law} for fixed magnitudes. 
	Note that because of the uniform distribution $f_\theta$ the \ac{cme} $\hat{\theta}$ is exactly the midpoint between the boundary angles which can also be seen later in Theorem \ref{prop:cme_multiple_pilots}.
	The optimization of the pilot sequence is thus given as 
	\begin{align}
		&\argmin_{\B a} ~\op E[(\theta - \hat{\theta}(\B a))^2] \\
		&= \argmin_{\B a} \int_0^{2\pi} (\theta - \hat{\theta}(\B a))^2f_{\theta}(\theta) \dx \theta \\
		&= \argmin_{\B a} \sum_{i=1}^{4M-1} \int_{\phi_i}^{\phi_{i+1}} (\theta - \hat{\theta}(\B a))^2 f_{\theta}(\theta) \dx \theta
		\label{eq:pilot_opt_argmin}
	\end{align}
	where in \eqref{eq:pilot_opt_argmin} the integral over the whole angle domain is split up into the $4M$ subsets that are defined via the boundary angles from \eqref{eq:p_r_given_phi_pilots} which fulfill $\phi_1 < \dots < \phi_{4M}$ where $\phi_1=0$ and $\phi_{4M} = 2\pi$.
	The optimization problem in \eqref{eq:pilot_opt_argmin} is identical to the well known problem of finding the quantization regions of a non-uniform quantizer which minimize the distortion \cite{gersho_quantization}. Because the angles are uniformly distributed, the optimal solution is found by equidistant integration regions, cf. \cite[Sec. 5.6]{gersho_quantization}, which is achieved by equidistant phase shifts in $[0,\frac{\pi}{2})$, i.e., $\psi_m=\frac{\pi(m-1)}{2M}~\forall m = 1,...,M$. Plugging this solution into \eqref{eq:pilot_vec_generic} results in \eqref{eq:pilot_vector} and finishes the proof.
\end{proof}

\subsection{Proof of Theorem \ref{prop:cme_multiple_pilots}}\label{app:proof_cme_multiple_pilots}
\begin{proof}
	We use the simplified representation of \eqref{eq:cme_pilots_split} where the mean of the Rayleigh distributed magnitudes is known to be, cf. \cite[Appendix A.1.3]{tse_viswanath_2005},
	\begin{equation}
		\op E [\alpha] = \frac{\sqrt{\pi}}{2}\sigma.
		\label{eq:cme_alpha}
	\end{equation}
Furthermore, due to the equidistant phase shifts which divide the input space into $4M$ distinct regions, cf. Appendix \ref{app:proof_opt_pilots}, and the symmetry of the Gaussian \ac{pdf}, we have the prior probability $\fr(\B r) = \frac{1}{4M}$ in this case. 
 Plugging this result and the definition of the uniform distribution $f_{\theta}$ into \eqref{eq:cme_phase_bayes}, the \ac{cme} of the normalized channel is computed to 
\begin{align}
	\op E[\op e^{\op j \theta } | \B r] &= \frac{4M}{2\pi} \int_{\phi_{\text{low}}(\B r)}^{\phi_{\text{low}}(\B r) + \frac{\pi}{2M}} \op e^{\op j \theta }  \dx\theta.
	\label{eq:cme_phase1}
\end{align}
It is sufficient to solve the integral for the first angle sector $[0,\frac{\pi}{2M})$ and afterwards shift the solution by the angle of $\phi_{\text{low}}(\B r)$. We thus compute
\begin{align}
	\int_{0}^{ \frac{\pi}{2M}} \op e^{\op j \theta} \dx\theta 
	&=  \sin\left(\frac{\pi}{2M}\right) + \op j\left(1 - \cos\left(\frac{\pi}{2M}\right)\right) \\
	&=  2\sin\left(\frac{\pi}{4M}\right)\exp\left(\op j\frac{\pi}{4M}\right).
	\label{eq:appA_integral}
\end{align}
After plugging \eqref{eq:cme_alpha} together with the generalized solution from \eqref{eq:cme_phase1} and \eqref{eq:appA_integral} into \eqref{eq:cme_pilots_split} yields \eqref{eq:cme_multiple_pilots} which finishes the proof.
\end{proof}

\subsection{Proof of Theorem \ref{prop:mse_equality}}\label{app:proof_mse_equality_pilots}

\begin{proof}
With the vector in \eqref{eq:pilot_vector}, its outer product is computed as
\begin{equation}
	[\B a \B a\h]_{m,n} = \exp\left(\op j \frac{\pi}{2M}(m-n)\right)
\end{equation}
and thus 
\begin{align}
	[\Re (\B a \B a\h)]_{m,n} &= \cos\left(\frac{\pi}{2M}(m-n)\right), \label{eq:re_aa} 
	\\
	[\Im (\B a \B a\h)]_{m,n} &= \sin\left(\frac{\pi}{2M}(m-n)\right). \label{eq:im_aa}
\end{align}
For the considered case, the unquantized observation is given as $\B y=\B a h$, and thus its covariance matrix can be written as $\B C_{\B y} = \sigma^2\B a \B a\h$. Furthermore, from the property that $|[\B a]_m|^2 = 1$, it follows that $\diag(\B C_{\B y}) = \sigma^2 \eye$. Consequently, the expression for the covariance matrix from \eqref{eq:arcsine_law} simplifies to
\begin{align}
	\B C_{\B r} &= \frac{2}{\pi}\left[\arcsin(\Re(\B a \B a\h))+\op j 
	\arcsin(\Im(\B a \B a\h))\right].
	\label{eq:Cr}
\end{align}
After plugging in \eqref{eq:re_aa} and \eqref{eq:im_aa} we get 
\begin{align}
	[\B C_{\B r}]_{m,n} &= \frac{2}{\pi}\left[\arcsin\left(\cos\left(\frac{\pi}{2M}(m-n)\right)\right)\right. 
	 \label{eq:Cr_middle123}
	\\ \nonumber 
	& ~~~\left.+\op j \arcsin\left(\sin\left(\frac{\pi}{2M}(m-n)\right)\right)\right]
	\\
	&= 1 - \frac{|m-n|}{M} + \op j \frac{m-n}{M},
	\label{eq:Cr_entry}
\end{align}
where the equation in \eqref{eq:Cr_entry} holds since the arguments in \eqref{eq:Cr_middle123} satisfy $-\frac{\pi}{2} < \frac{\pi}{2M}(m-n) < \frac{\pi}{2}$.
The following Lemma computes $\B C_{\B r}\inv$ in closed form.
\begin{lemma}\label{prop:Cr_inv}
	The inverse of $\B C_{\B r}$ from \eqref{eq:Cr_entry} is given as 
	\begin{equation}
		[\B C_{\B r}\inv]_{m,n} = \begin{cases}
			M  & m=n\\
			-\frac{M}{2} & |m-n| = 1 \\
			\op j\frac{M(n-m)}{2(M-1)}  & |m-n| = M-1 \\
			0 &\text{otherwise}
		\end{cases}.
		\label{eq:Cr_inv}
	\end{equation}
\end{lemma}
\begin{proof}
	To prove that \eqref{eq:Cr_inv} is the inverse of \eqref{eq:Cr} we define the matrix $\B T = \B C_{\B r}\inv \B C_{\B r}$. If we show that $\B T$ is the identity matrix, i.e., $\B T = \eye$, then Lemma \ref{prop:Cr_inv} holds.
	The $m,n$th element of $\B T$ is given as $[\B T]_{m,n}  = \sum_{l=1}^M [\B C_{\B r}\inv]_{m,l}[\B C_{\B r}]_{l,n} $.
	Due to the sparsity of $\B C_{\B r}\inv$, only three summands are non-zero. 
	 For example, we can compute \( [\B T]_{m,n} \) for \( m = n = 1 \):
	\begin{align*}
		&[\B T]_{1,1}
		= \sum_{l=1}^M [\B C_{\B r}^{-1}]_{1,l} [\B C_{\B r}]_{l,1}
		\\&= [\B C_{\B r}^{-1}]_{1,1} [\B C_{\B r}]_{1,1}
		+ [\B C_{\B r}^{-1}]_{1,2} [\B C_{\B r}]_{2,1}
		+ [\B C_{\B r}^{-1}]_{1,M} [\B C_{\B r}]_{M,1}
		\\&= M \cdot 1
		+ \left( -\frac{M}{2} \right) \cdot \left( 1 - \frac{1}{M} + \op j \frac{1}{M} \right)
		\\& \quad+ \op j \frac{M}{2} \cdot \left( 1 - \frac{M-1}{M} + \op j \frac{M-1}{M} \right) = 1.
	\end{align*}
	We show that the diagonal elements of $\B T$, i.e., $m=n$, are one in \eqref{eq:Cr_inv_diag} and that all off-diagonal elements of $\B T$, i.e., $m\neq n$, are zero in \eqref{eq:Cr_inv_offdiag}. 
\end{proof}

	\setcounter{equation}{68}
With Lemma \ref{prop:Cr_inv}, we can now find the expression of 
\begin{align}
	\B a\h \B C_{\B r }\inv \B a &= 
		\sum_{m=1}^M\sum_{n=1}^M [\B a\h]_m [\B C_{\B r}\inv]_{m,n}[\B a]_n \\
		&= \sum_{m=1}^M M - \frac{M}{2} \sum_{m=1}^M \left( \exp(-\op j \frac{\pi}{2M}) +  \exp(\op j \frac{\pi}{2M})\right) \\
		&= M^2 - M^2 \cos\left(\frac{\pi}{2M}\right) \\
		&= 2M^2\sin^2\left(\frac{\pi}{4M}\right).
\end{align}
Plugging this into \eqref{eq:mse_buss_multiple_pilots} results in the same \ac{mse} term as in \eqref{eq:mse_cme_multiple_pilots} and thus finishes the proof.
\end{proof}

\balance
\bibliographystyle{IEEEtran}
\bibliography{IEEEabrv,library}

\end{document}